\documentclass[10pt, conference, letterpaper]{IEEEtran}

\usepackage{amsmath,amssymb,amsthm}
\usepackage{algorithm,algcompatible}
\usepackage{array}
\usepackage{enumerate}
\usepackage[pdftex]{graphicx}
\usepackage{url}
\newtheorem{mylem}{Lemma}
\newtheorem{mypro}{Proposition}

\begin{document}
\title{A Low Complexity Detection Algorithm for SCMA}

\author{\IEEEauthorblockN{Chenchen Zhang}
\IEEEauthorblockA{Department of Computer\\Science and Engineering\\Shanghai Jiao Tong University\\
China \\Email:kevinzhangcc@sjtu.edu.cn}

\and
\IEEEauthorblockN{Yuan Luo}
\IEEEauthorblockA{Department of Computer\\Science and Engineering\\Shanghai Jiao Tong University\\
China\\
Email:yuanluo@sjtu.edu.cn}

\and
\IEEEauthorblockN{Yan Chen}
\IEEEauthorblockA{Huawei Technologies, Co.Ltd., Shanghai \\China\\
Email:bigbird.chenyan@huawei.com}}

\maketitle

\begin{abstract}
Sparse code multiple access (SCMA) is a new multiple access technique which supports massive connectivity. Compared with the current Long Term Evolution (LTE) system, it enables the overloading of active users on limited orthogonal resources  and thus meets the requirement of the fifth generation (5G) wireless networks. However, the computation complexity of existing detection algorithms increases exponentially with $d_f$ (the degree of the resource nodes). Although the codebooks are designed to have low density, the detection still takes considerable time. The parameter $d_f$ must be designed to be very small, which largely limits the choice of codebooks. In this paper, a new detection algorithm is proposed by discretizing the probability distribution functions (PDFs) in the layer nodes (variable nodes). Given $M$ as the size of one codebook, the detection complexity of each resource node (function node) is reduced from $O(d_f M^{d_f})$ to $O(d_f^3 \ln (d_f))$. Its detection accuracy can quickly approach that of the previous detection algorithms with the decrease of sampling interval in discretization.
\end{abstract}


% \fi
\IEEEpeerreviewmaketitle

\section{Introduction}
The upcoming 5G wireless networks are supposed to support at least 100 billion devices, for which one fundamental requirement is the capability for massive connectivity \cite{hwWhitepaper}. Since the spectrum resource is limited, it is hard for current LTE systems to handle so many active users simultaneously. Hoshyar et al. have proposed an overloaded system called low-density signature (LDS) \cite{LDS} which can serve more active users than the number of orthogonal resources. In LDS, each user's bits will spread over the orthogonal resources according to his signature. To keep the detection complexity feasible, the signature should contain only a few number of nonzero elements. Based on LDS, SCMA \cite{SCMA} was proposed with extra coding gain, which results in better performance than LDS. With its technical advantage, SCMA is very likely to play an important role in 5G wireless networks.

SCMA prefers multi-dimensional codebooks rather than signatures to non-orthogonally spread the transmitted bits. The access points for users in SCMA are called layers and each layer has a codebook. When a user transmits some bits on one layer, the bits are mapped into multi-dimensional complex codewords according to the corresponding codebook. Different layers can work simultaneously and their codewords are superposed in the time domain. Comparing with the QAM modulation in LDS, the multi-dimensional modulation in SCMA is more efficient to utilize the constellation, which brings coding gains \cite{codingGain}.

The performance of SCMA is mainly determined by the codebooks \cite{cddesign2} and the detection algorithm. The complexity of the detection algorithm is very important because it not only determines the delay in detection but also affects the choice of codebooks. Until now, the most efficient detection method is the message passing algorithm (MPA) \cite{SCMA} which takes advantage of the low density of the codebooks. By using MPA, the receiver can perform a maximum a posterior (MAP) detection \cite{sumProd} with computation complexity $O(d_f M^{d_f})$ per resource node. There are already some optimizations on the detection \cite{LDS, simplifiedLLR}, but the complexity is still exponential.

It may seem hard to reduce the order of detection complexity. In each iteration, the message from one resource node to one layer node is computed depending on the messages from the rest $d_f-1$ layer nodes which are connected to that resource node. As each layer has $M$ possible values, there are $M^{d_f-1}$ combinations need to be considered. In the previous works, the optimizations mainly focus on how to deal with each combination quickly. For example, the log-likelihood ratio (LLR) \cite{LDS, llr2} replaces the multiplication operations by addition operation.

In this paper, the update of messages is investigated from a new perspective. Actually, each layer node is not treated as $M$ candidate values but a random variable on complex field. The PDFs of the $d_f-1$ random variables are discretized and convolved to get the new message. As a result, the complexity to update one message becomes the complexity of the 2-dimensional fast fourier transform (2-D FFT) algorithm, which is $O(d_f^2 \ln (d_f))$ as the involved region grows linearly with $d_f$. When the real part in the system is independent of the imaginary part, the complexity order to update one message is further reduced to $d_f \ln (d_f)$, and the original MPA is reduced to $M^{d_f/2}$ \cite{codebook}. Note that the increase of codebook size $M$ won't largely affect the time consumption in the proposed method. For sake of simplicity, hereinafter we use discretized message passing algorithm (DMPA) to denote the  proposed algorithm.

This paper is organized as follows. Section \ref{II} gives the motivation of this work and the SCMA system model. Then the DMPA is proposed in Section III. Firstly the original MPA in SCMA is briefly described. Secondly it is compared with the MPA in LDPC codes and demonstrated in another perspective. Thirdly the detailed procedure of DMPA is presented based on the demonstration. In Section IV, we give the theoretical analysis for DMPA on computation complexity and detection accuracy, which shows the advantages of the new method. The theoretical analysis is confirmed by simulations in Section V. Finally Section VI concludes this paper. We use $x$, $\textbf{x}$ and $\mathcal{X}$ to represent a scalar, a vector and a set, respectively.

\section{Preliminaries}\label{II}
%In this section, the motivation of this work is presented in Subsection \ref{II1}. Then Subsection \ref{II2} describes the SCMA system model in details. We use $x$, $\textbf{x}$ and $\mathcal{X}$ to represent a scalar, a vector and a set.

\subsection{Motivation}\label{II1}
%Comparing with previous multiple access techniques, SCMA has many advantages but also some shortcomings. One of the most important shortcomings is that the detection has a high complexity.
Comparing with orthogonal multiple access techniques, the non-orthogonal techniques can support much more active users simultaneously but always need detection to separate different users' messages. Because of the increasing requirement for massive connectivity and limited spectrum resources, SCMA, a typical non-orthogonal technique, is very likely to become one key technique in the future 5G systems. However, the high detection complexity will increase the latency and limit the design of codebooks, which may be an important bottleneck.

SCMA and LDPC codes are similar in factor graph representation and the corresponding MPA detection. In each resource node (or check node), to update the message to one variable node they both need to use the messages from the rest connected variable nodes. However, the complexity of MPA in LDPC codes is far more less than that in SCMA. The reason is that in LDPC codes the new message can be computed by FFT (See Subsection III-B in \cite{ldpc}) while it is near brute-force in SCMA. Inspired by LDPC codes, we believe that there should be some methods far more efficient than the exponential time complexity method. The problem is that the method in LDPC codes cannot be directly applied on SCMA as the variable nodes in \cite{ldpc} only have two possible values and the operations are all on $\mathbb{F}_2$ field. This paper solves the problem by discretizing the PDFs of variable nodes, which largely reduces the detection complexity of SCMA.

\subsection{System Model}\label{II2}
An SCMA system consists of its encoder, channel and receiver. The transmitted bits are encoded into SCMA codewords in encoder and multiplexed in the channel. The receiver near-optimally detects the SCMA codewords by MPA over the corresponding factor graph.

The SCMA encoder contains $J$ separate layers and each layer has a codebook $\mathcal{X}_j$. We assume every codebook contains $M$ complex codewords, i.e., $\mathcal{X}_j = \{\textbf{x}_{j1},\ldots,\textbf{x}_{jM}\}$. Let $\{x_i\}_{i=1}^I$ denote the vector $(x_1,\ldots,x_I)^T$, then the $m$th codeword can be written as $\textbf{x}_{jm} = \{x_{kjm}\}_{k=1}^K$. Although the length of the codewords is $K$, each codebook has $K-N$ positions on which its codewords are all zero. In each layer, every $\log_2 M$ bits will be mapped into one codeword. The $J$ layers can work simultaneously, which produces $J$ codewords.

Corresponding with the length of the codewords, the SCMA channel has $K$ orthogonal resources such as OFDMA tones or MIMO spatial layers. The $J$ codewords from the layers are multiplexed over the $K$ resources and the $j$th codeword will be affected by the channel vector $\textbf{h}_j=\{h_{kj}\}_{k=1}^K$. Let $diag(\textbf{h}_j)$ denote the diagonal matrix where its $n$th diagonal element is the $n$th element of $\textbf{h}_j$. If the channel is a complex additive white Gaussian noise (AWGN) channel, the received signal $\textbf{y}$ can be expressed as
\begin{equation}\label{scma}
  \textbf{y}=\sum_{j=1}^J diag(\textbf{h}_j)\textbf{x}_j+\textbf{n},
\end{equation}
where $\textbf{x}_j=(x_{1j},\ldots,x_{Kj})^T$ is the $j$th codeword and $\textbf{n}\sim\mathcal{CN}(0,N_0 \textbf{I})$ is the complex Gaussian noise. The overloading factor is defined as $\lambda := J/K$.

\begin{figure}[htbp]
  \centering
  \includegraphics[width=0.4\textwidth]{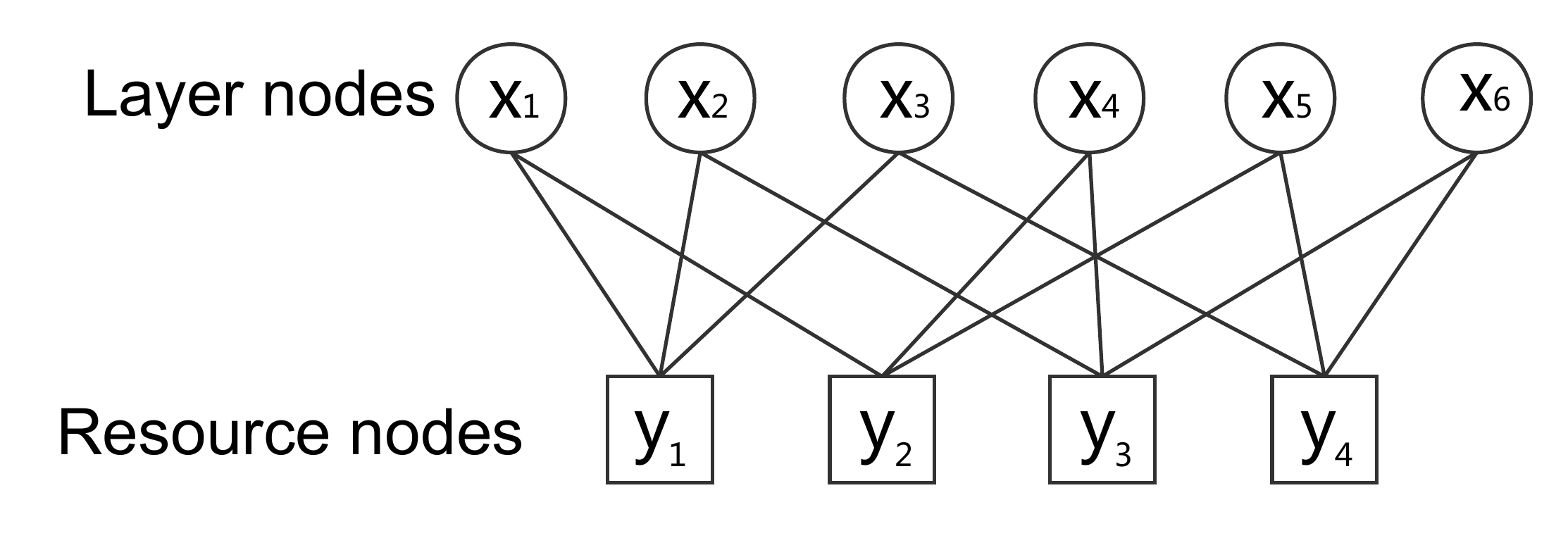}
  \centering
  \caption{Factor graph representation of an SCMA system with $J=6$, $K=4$, $N=2$ and $d_f=3$}
  \label{scmaFG}
\end{figure}

The receiver detects the transmitted codewords according to the received signal $\textbf{y}$ and the channel knowledge. Up to now, the most efficient detector is the MPA detector which is based on factor graph. As is shown in Fig. \ref{scmaFG}, the factor graph is formed by the relation between layers and resources. In each codebook, the codewords are only nonzero on $N$ positions. It means that bits from each layer will only spread in $N$ resources. In this paper, we assume the $K$ resource nodes in the factor graph have the same degree $d_f$. Under this assumption the update of each message from resource nodes has the complexity order $M^{d_f}$. Since the degree of resource nodes is $d_f$, the complexity in each resource node is $O(d_f M^{d_f})$.

\section{Discretized Message Passing Algorithm}\label{III}
This section proposes a new detection method DMPA which reduces the complexity per resource node from $O(d_f M^{d_f})$ to $O(d_f^3 \ln (d_f))$. Subsection \ref{III1} briefly describes the original MPA detector for reference. Then it is compared with MPA in LDPC codes in Subsection \ref{III2}. Finally Subsection \ref{III3} presents the basic idea and detailed procedure of DMPA.

\subsection{Original MPA for SCMA}\label{III1}
There are already some optimizations on the MPA detector for SCMA. But as the original one can better reflect its basic idea, we choose the original MPA (See Algorithm 1 in \cite{simplifiedLLR}) as the referenced method.

In SCMA, each codebook contains $M$ codewords and everyone has a probability to be the transmitted codeword. The probability distribution is computed as the messages in MPA according to the received signal $\textbf{y}$, the channel knowledge, and the factor graph. Actually, the structures of factor graphs represent
some constraints and the messages are updated based on the constraints. For example, in the system depicted in Fig. \ref{scmaFG}, the three edges connected with the first resource node implies the constraint $h_{11}x_1+h_{21}x_2+h_{31}x_3+n=y_1$ where $h_{i1}$ is an element of channel vectors and $x_i$ is the first element of the coresponding codeword which has $M$ possible values $\{x_{im}\}_{m=1}^M$. In each iteration the MPA detector will deduce the probability distribution of $x_1$ according to $x_2, x_3$ and the value of $y_1$. Let $[N]$ denote the set $\{1,2,\ldots,N\}$, the detailed procedure is shown in Algorithm \ref{mpaProcedure}.

\begin{algorithm}
\caption{MPA detection for SCMA}
\label{mpaProcedure}
\begin{algorithmic}
\REQUIRE $\textbf{y}$, $\{\textbf{h}_j\}_{j=1}^{J}$, $\{\mathcal{X}_j\}_{j=1}^{J}$, $N_0$
\ENSURE
%\begin{equation}\label{init}
  $V_{j\rightarrow k}^{(0)}(\textbf{x}_{jm})=1/M, ~~(j\!\in\![J],k\!\in\![K],m\!\in\![M])$
%\end{equation}
\FOR {$t=1$ to itrNum}
%\algstore{myalg1}
%\end{algorithmic}
%\end{algorithm}
%
%
%\begin{algorithm}
%\begin{algorithmic}
%\algrestore{myalg1}
\STATE\COMMENT {\textbf{update messages from resource nodes}}
\FORALL{$j\in [J],k\in [K],m\in [M]$ such that Edge(j,k) exists in factor graph}
\STATE
\begin{align}\label{FNupdate}
  U_{k\rightarrow j}^{(t)}(\textbf{x}_{jm})=\sum_{\textbf{c}\in com}\frac{1}{\pi N_0} exp[-\frac{1}{N_0} |y_k-h_{kj}x_{kjm} &  \notag \\
  -\sum_{i\in \partial k\backslash j} h_{ki}c_{ik}|^2] \prod_{i\in \partial k\backslash j} V_{i\rightarrow k}^{(t-1)}(\textbf{c}_i). &
\end{align}

$com = \mathcal{X}_{i_1} \times \ldots \times \mathcal{X}_{i_{d_f-1}}$ is the Cartesian Product of the codebooks of the rest $d_f-1$ variable nodes connected to the $k$th function node, i.e., $i_{\ast} \in\partial k\backslash j$. $\textbf{c}=(\textbf{c}_{i_1},\ldots, \textbf{c}_{i_{d_f-1}})$ is one element of $com$ and $\textbf{c}_i=\{c_{ik}\}_{k=1}^K$.
\ENDFOR
\STATE\COMMENT {\textbf{update messages from layer nodes}}
\FORALL{$j\in [J],k\in [K],m\in [M]$ such that Edge(j,k) exists in factor graph}
\STATE
\begin{equation}\label{VNupdate}
  V_{j\rightarrow k}^{(t)}(\textbf{x}_{jm})=\prod_{l\in\partial j\backslash k} U_{l\rightarrow j}^{(t)}(\textbf{x}_{jm}).
\end{equation}
\ENDFOR

Normalize the probabilities to keep numerical stable:
\FORALL{$j\in [J],k\in [K],m\in [M]$ such that Edge(j,k) exists in factor graph}
\STATE
\begin{equation}\label{VNnormalize}
V_{j\rightarrow k}^{(t)}(\textbf{x}_{jm})=V_{j\rightarrow k}^{(t)}(\textbf{x}_{jm})/\sum_{m'=1}^M V_{j\rightarrow k}^{(t)}(\textbf{x}_{jm'})
\end{equation}
\ENDFOR
\ENDFOR
\STATE\COMMENT {\textbf{make decision after some iterations}}
\FOR {$j=1,\ldots,J; m=1,\ldots, M$}
\STATE
\begin{equation}
  V_j(\textbf{x}_{jm})=\prod_{k\in\partial j}U_{k\rightarrow j}^{(itrNum)}(\textbf{x}_{jm}).
\end{equation}
\ENDFOR
\STATE
Finally, the one in $\{\textbf{x}_{jm}\}_{m=1}^M$ which maximizes $V_j(\cdot)$ is regarded as the transmitted codeword in the $j$th layer.
\end{algorithmic}
\end{algorithm}

\subsection{Comparison with MPA decoding of LDPC codes}\label{III2}
This subsection first analyzes the underlying idea of Algorithm \ref{mpaProcedure}. Then the idea is compared with the low compelxity MPA in binary LDPC codes, which inspires us to propose the DMPA in Subsection \ref{III3}.

In Algorithm \ref{mpaProcedure}, the dominant cost operation is to compute Equation (\ref{FNupdate}), which calculates the probability distribution of the $j$th layer node according to the $k$th resource node and the rest $d_f-1$ adjacent layer nodes. In Equation (\ref{FNupdate}), the first part $\frac{1}{\pi N_0} exp[\cdot]$ computes the probability that the transmitted codeword in $j$th layer is $\textbf{x}_{jm}$ in the condition that the rest $d_f-1$ layers are according to $\textbf{c}$, i.e., $P(\textbf{x}_{jm}|\textbf{c})$. The last part $V_{i\rightarrow k}^{(t-1)}(\textbf{c}_i)$ is the probability that the codeword from $i$th layer is $\textbf{c}_i$, i.e., $P(\textbf{c}_i)$. The product of the two parts is the joint probability $P(\textbf{x}_{jm}, \textbf{c})$. So the whole equation is a traversal on the set $com$ to compute $P(\textbf{x}_{jm})$.

The traversal is near brute-force, which causes exponential complexity. The decoding of LDPC codes computes probability distributions in a much smarter way. In binary LDPC codes, the constraint $x_1+x_2+x_3=0$ can lead to $x_1=x_2+x_3$, which means the PDF of $x_1$ is the convolution of PDFs of $x_2$ and $x_3$. As $x_i$ only has two possible values $0$ and $1$, the convolution of $x_1$ can be calculated by Fourier transform on $\mathbb{F}_2$ field \cite{ldpc}. In SCMA, the constraint is like $h_{11}x_1+h_{21}x_2+h_{11}x_3+n=y_1$, where $n$ is the noise and $y_1$ is a constant. It seems that the probability distribution of $x_1$ can be obtained by computing convolution in a similar way.
%However, since the sample spaces of the random variables are different and the operations are on complex field, the convolution cannot be solved by Fourier transform directly.
However, if $x_i$ is seen as a random variable whose sampling space consists of $M$ individual numbers, the Fourier transform cannot be directly applied to the convolution since the numbers from different layers are even not aligned. Our key point to solve the problem is to regard the layer nodes as random variables on the whole complex field and do discretization (See Subsection \ref{III3}).

\subsection{DMPA}\label{III3}
This subsection gives the procedure to solve Equation (\ref{FNupdate}) efficiently. Firstly, (\ref{FNupdate}) is proved to be equivalent to computing a convolution. Secondly, the convolution is solved by discretization and FFT. Finally, the DMPA on real field is extended to the complex field.

 To better demonstrate our basic idea, we assume the elements of channel vectors are all $1$ and the real part of codebooks is designed independent with the imaginary part. Under this assumption, the real part and the imaginary part in the SCMA system are independent. Therefore we further assume the random variables are all on real field \cite{codebook}. If $f(t)$ is a function, let $f$ denote its sampling sequence and the $f[n]$ is the $n$th element in the sequence.

\textbf{Equation (\ref{FNupdate}) is equivalent to the convolution of $d_f$ PDFs at the point $x_{kjm}$}. For simplicity, assume the constraint to update message is $x_1+x_2+\ldots+x_{d_f}+n=y$, and the corresponding PDFs are $f_1(t),\ldots,f_{d_f}(t),\eta(t)$. The PDF $f_i(t)$ is all zero except $M$ impulses, i.e.,
\begin{equation}\label{varpdf}
  f_i(t)=\sum_{m=1}^M P(x_i=x_{im})\delta(t-x_{im}).
\end{equation}
 Note that the convolution of the impulse $\delta(t-a)$ with any function $f(t)$ is a shift on $f(t)$, i.e., $f(t-a)$.
 %It is intuitive to get the PDF of $x_2+\ldots+x_{d_f}$ as shown in Fig. \ref{conv1}.
 Let $g(t)$ denote the PDF of $x_2+\ldots+x_{d_f}+n$. Then it follows that
 \begin{align}\label{pdfOFg1}
   (f_2*\ldots*f_{d_f})(t)&= \notag \\
   \sum_{m_2,\ldots,m_{d_f}  \in [M]}& \left[\delta(t-\sum_{i=2}^{d_f}x_{im_i}) \prod_{i=2}^{d_f}P(x_i=x_{im_i})\right],
 \end{align}
 and
 \begin{align}\label{pdfOFg2}
   g(t)=&(f_2*\ldots*f_{d_f})(t)*\eta(t)\notag\\
   =& \sum_{m_2,\ldots,m_{d_f} \in [M]} \left[\eta(t-\sum_{i=2}^{d_f}x_{im_i}) \prod_{i=2}^{d_f}P(x_i=x_{im_i})\right].
 \end{align}
 Equation (\ref{pdfOFg1}), (\ref{pdfOFg2}) shows that $(f_2*\ldots*f_{d_f})(t)$ is the summation of $M^{d_f-1}$ impulses and $g(t)$ is the superposition of $M^{d_f-1}$ noise PDFs.
 \begin{figure}[htbp]
  \centering
  \includegraphics[width=0.5\textwidth]{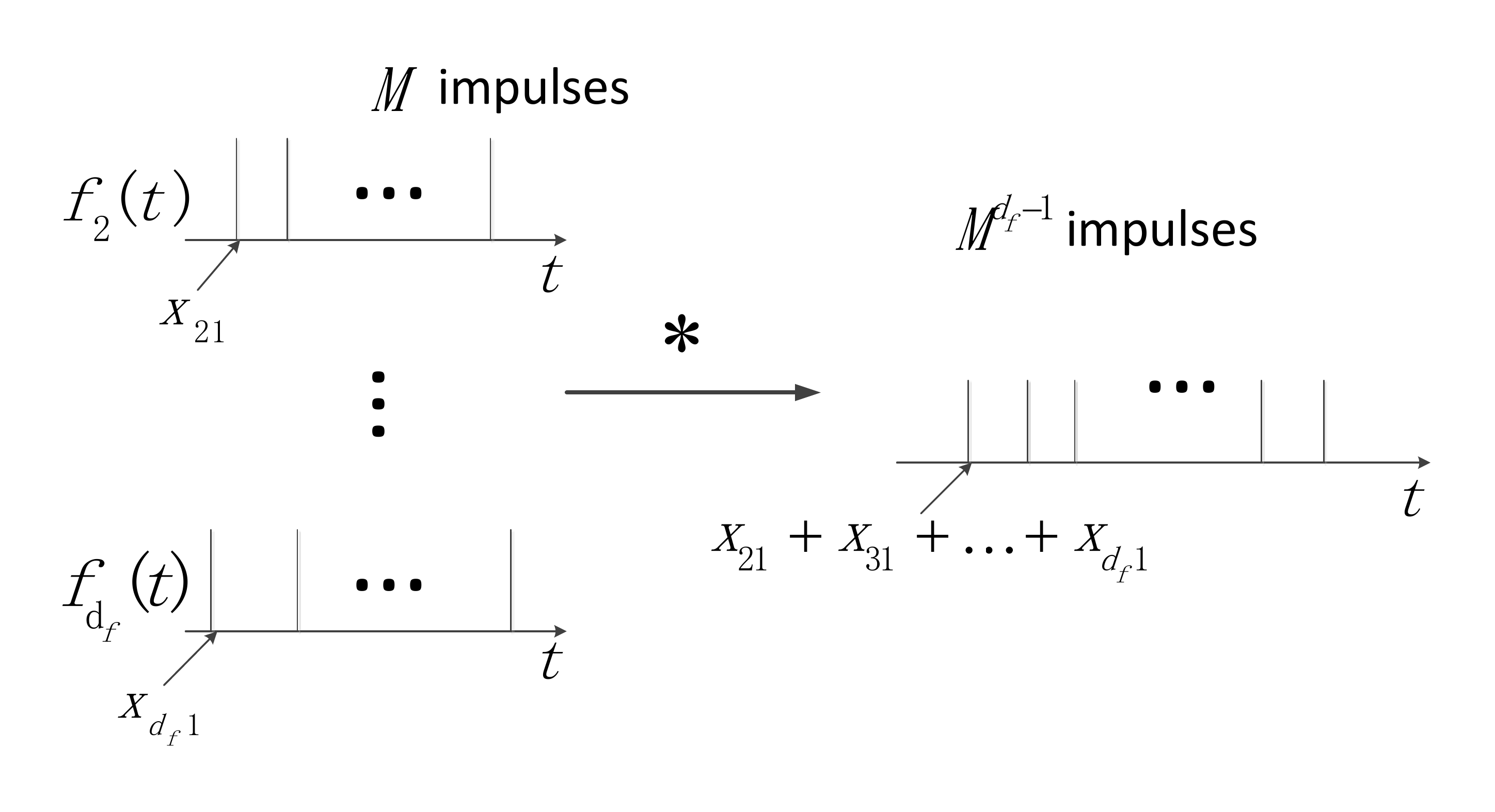}
  \centering
  \caption{The convolution of impulses}
  \label{conv1}
\end{figure}

Fig. \ref{conv1} visually describes the convolution in (\ref{pdfOFg1}). Its complexity can be largely reduced under some conditions. Suppose the impulses in the left are all integral multiples of $0.01$, then the $M^{d_f-1}$ impulses in the right will overlap on the integral multiples of $0.01$ within a interval whose length increases linearly with $d_f-1$. The DMPA first shifts impulses to the sampling positions by discretization, then it only needs to deal with the sampling points in the interval instead of the $M^{d_f-1}$ impulses. Since the number of sampling points grows linearly with $d_f-1$, the complexity is largely reduced.
 \begin{figure}[htbp]
  \centering
  \includegraphics[width=0.5\textwidth]{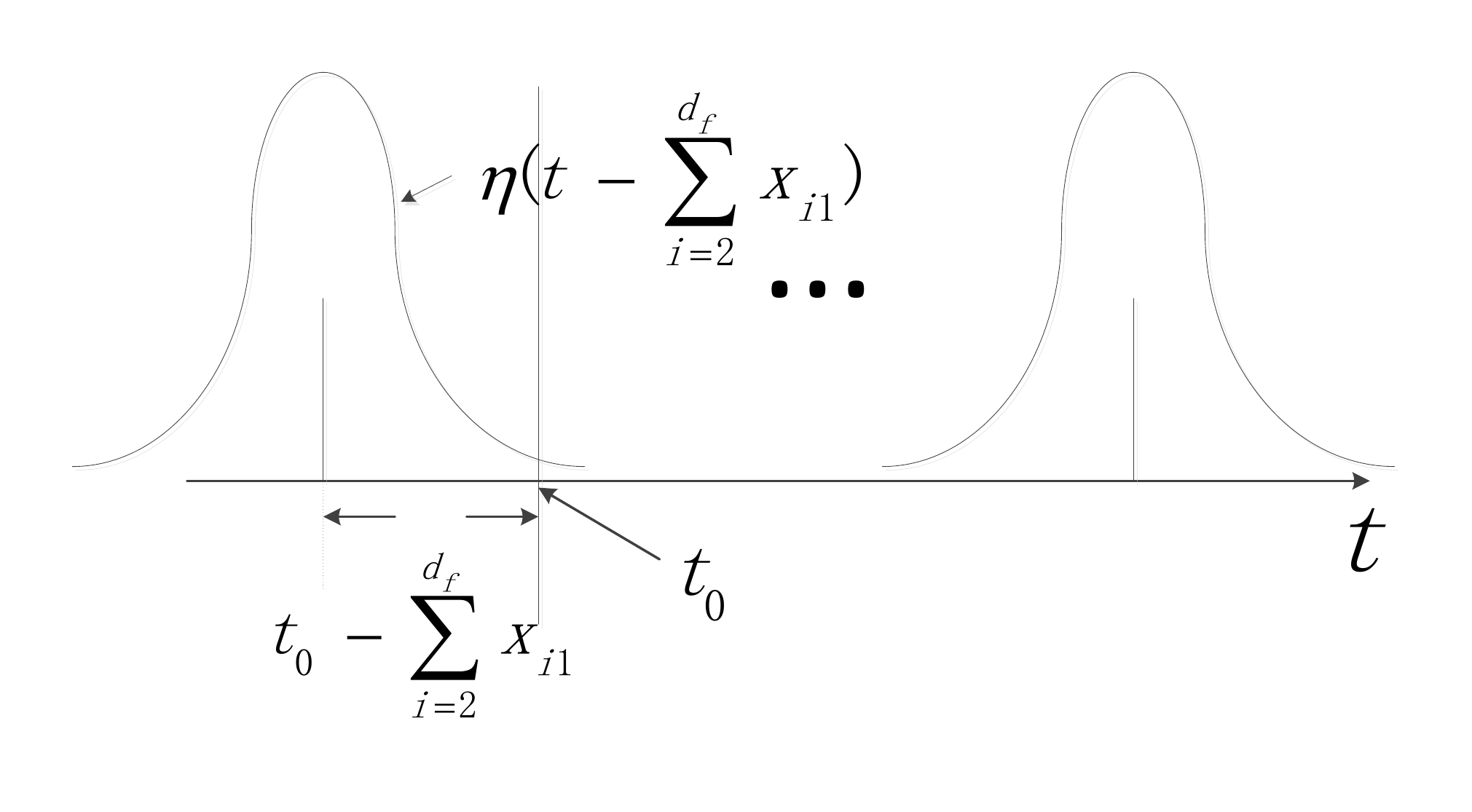}
  \centering
  \caption{The convolution of impulses and noise}
  \label{conv2}
\end{figure}

In Equation (\ref{pdfOFg2}), ``$m_2,\ldots,m_{d_f} \in [M]$" is actually a traversal on the set $com$. The $\eta(t-\sum_{i=2}^{d_f}x_{im_i})$ and $\prod_{i=2}^{d_f}P(x_i=x_{im_i})$ are corresponding with the $\frac{1}{\pi N_0} exp[\cdot]$ and $\prod_{i\in \partial k\backslash j} V_{i\rightarrow k} ^{(t-1)} (\textbf{c}_i)$ in (\ref{FNupdate}), respectively. When the $f_1(x)$ is updated according to the constraint, it can be expressed as
 \begin{align}\label{eqv}
   f_1(x_{k1m})=& g(y_k-x_{k1m}) \notag \\
   =&\!\sum_{\textbf{c}\in com}\!\eta(y_k\!-\!x_{k1m}\!-\!\sum_{i\in \partial k\backslash 1}\!c_{ik})\cdot\!\prod_{i\in \partial k\backslash 1}\!V_{1\rightarrow k} ^{(t-1)}(\textbf{c}_i).
 \end{align}
  Equation (\ref{eqv}) is the same with Equation (\ref{FNupdate}) with all $1$ channel vectors, and it can be easily extended to the cases with general channel vectors.  Therefore the Equation (\ref{FNupdate}) is actually to compute the convolution in $x_{kjm}$, i.e., $U_{k\rightarrow j}^{(t)}(\textbf{x}_{jm})=f_j(x_{kjm})$.

 \textbf{The convolution can be efficiently computed by discretization and FFT}. Still consider the constraint $x_1+x_2+\ldots+x_{d_f}+n=y$.  Suppose the possible values of $x_j$ ($j\in[d_f]$) are all in the interval $[-wid, wid]$ and $\eta(t)$ is negligible outside $[-nWid, nWid]$. Let $w$ denote the precision of discretization which is an exact division of $wid$ and $nWid$. The discretization for $f_i(x)$ and $\eta(x)$ goes as follows:
\begin{equation}\label{disc1}
    f_i[n]=\int_{-wid+(n-\frac{1}{2})w}^{-wid+(n+\frac{1}{2})w} f_i(x) dx,
\end{equation}
where $n=\{0,1,\ldots,(2Wid)/w\}$, and
\begin{equation}\label{disc2}
    \eta[n]=\eta(-nWid+n\cdot w),
\end{equation}
where $n=\{0,1,\ldots,(2nWid)/w\}$. Equation (\ref{disc1}) means that the integral of each impulse is distributed to the nearest sampling point, and Equation (\ref{disc2}) is simply sampling on the PDF of noise. After the discretization, the impulses in Fig. \ref{conv1} will be all on the sampling positions. This step brings detection accuracy loss which is analyzed in Subsection \ref{IV2}.

The linear discrete convolution of two discrete sequence, such as $f_1$ and $f_2$, is defined as
\begin{equation}
  (f_1*f_2)[n]=\sum_{k=-\infty}^{\infty}f_1[k]f_2[n-k],
\end{equation}
where $n=\{0,1,\ldots,(4wid)/w\}$ and the excess index gets $0$. In linear discrete convolution, every non-zero bit of $f_1$ will obtain a shift of $f_2$ and the whole result $f_1*f_2$ is actually the superposition of several shifted $f_2$, which is a good approximation of the continuous convolution with impulses. The accuracy of this approximation is analyzed in Subsection \ref{IV2}. It is obvious that the linear discrete convolution cannot be computed by FFT directly since the length of $f_1,f_2$ is different with the length of $f_1*f_2$. However, the linear discrete convolution is easy to compute using circular discrete convolution, which can be computed by FFT. Let $f_2'$ denote the periodic summation of $f_2$, i.e.,
%\begin{equation}
%  f_2'[n] \overset{\text{def}}{=} \sum_{p=-\infty}^{\infty} f_2[n+p\cdot length(f_2)]=f_2[n(mod length(f_2))].
%\end{equation}
\begin{align}
   f_2'[n] \overset{\text{def}}{=} & \sum_{p=-\infty}^{\infty} f_2[n+p\cdot length(f_2)] \notag \\
   =&f_2[n\mod length(f_2)].
\end{align}
Then the circular discrete convolution of $f_1$ and $f_2$ can be expressed as
\begin{equation}
  (f_1*f_2')[n]=\sum_{k=-\infty}^{\infty}\left( f_1[k]\sum_{p=-\infty}^{\infty} f_2[n\!-\!k\!+\!p\!\cdot\!length(f_2)]\right),
\end{equation}
where $n=\{0,1,\ldots,(2wid)/w\}$. By padding $(2wid)/w$ zeros to $f_1$ and $f_2$, the result of circular convolution is the same with the linear convolution.
\begin{mylem}[Circular convolution theorem]
  The circular convolution of two finite sequences with same length can be obtained as the inverse Fourier transform of the Hadamard product of the individual Fourier transforms.
\end{mylem}
\noindent According to the circular convolution theorem, the sampling sequence of $g(t)$ can be approximately computed by
\begin{equation}\label{convAll}
  g=\mathcal{F}^{-1}[\mathcal{F}(pad(f_2))\bullet \cdots \bullet \mathcal{F}(pad(f_{d_f}))\bullet \mathcal{F}(pad(\eta))],
\end{equation}
where $\bullet$ denote the Hadamard product. The function value $g(t_0)$ is regarded as the value of the nearest sampling point, i.e.,
\begin{equation}\label{discToFunc}
  g(t_0)=g[round((t_0+wid(df-1)+nWid)/w)],
\end{equation}
where $round$ is the matlab function to get the nearest integer.

The previous two paragraphs shows that the Equation (\ref{FNupdate}) can be solved approximately by FFT. Since the complexity of other parts in Algorithm \ref{mpaProcedure} is negligible compared with the Equation (\ref{FNupdate}), the proposed method largely reduces the whole complexity. The procedure of the DMPA is different with MPA only in the part that updates messages from resource nodes, as shown in Algorithm \ref{dmpaProcedure}.
\begin{algorithm}
\caption{DMPA detection for SCMA}
\label{dmpaProcedure}
\begin{algorithmic}
\STATE\COMMENT {\textbf{The previous part is the same with Algorithm \ref{mpaProcedure}}}
\STATE\COMMENT {\textbf{update messages from resource nodes}}
\FORALL{$j\in [J],k\in [K]$ such that $Edge(j,k)$ exists}
\STATE
\textbf{1}. Get $\{f_{i_m}\}_{m=1}^{d_f-1}(t)$ from $V_{i\rightarrow k}^{(t-1)}(\cdot)$ according to (\ref{varpdf}).
\STATE
\textbf{2}. Discretize $\{f_{i_m}\}_{m=1}^{d_f-1}(t)$ and $\eta(t)$ according to (\ref{disc1}) and (\ref{disc2}).
\STATE
\textbf{3}. Pad zeros to $\{f_{i_m}\}_{m=1}^{d_f-1}$ and $\eta$ to obtain sequences with length $N$ as follows:

%\algstore{myalg2}
%\end{algorithmic}
%\end{algorithm}
%
%
%\begin{algorithm}
%\begin{algorithmic}
%\algrestore{myalg2}
\STATE
\begin{equation}\label{discLen}
  N=pow2(nextpow2(2((d_f-1)wid+nWid)/w+1)),
\end{equation}
where the $pow2$ and $nextpow2$ are matlab functions.
\STATE
\textbf{4}. Get sequence $g$ according to (\ref{convAll}).
\FOR{$m\in[M]$}
\STATE
Compute Equation (\ref{FNupdate}) according to (\ref{discToFunc}):
\begin{equation}\label{appFNupdate}
  U_{k\rightarrow j}^{(t)}(\textbf{x}_{jm})=g(y_k-x_{kjm}).
\end{equation}
\ENDFOR
\ENDFOR
\STATE\COMMENT {\textbf{The following part is the same with Algorithm \ref{mpaProcedure}}}
\end{algorithmic}
\end{algorithm}

The analysis in this subsection is based on the assumption that the random variables are all on real field, but it is easy to generalize them to complex field. Without the assumption, $f_i(t)$ in (\ref{varpdf}) becomes $M$ impulses on complex field and the noise PDF $\eta(t)$ becomes the complex Gaussian distribution in AWGN channels. Their discretization are 2-D matrix which can be handled by the 2-D FFT. Besides, we assumed that the elements of the channel vectors are all $1$. Note that a general system can be seen as having all $1$ channel vectors when $\{\textbf{h}_{j}\bullet\textbf{x}_{jm}\}_{j\in[J],m\in[M]}$ are regarded as the codewords. The DMPA is easy to extended to general cases.

\section{Performance Analysis}\label{IV}
The idea and detailed procedure of DMPA are presented in Section \ref{III}. However, it still needs to demonstrate that whether the discretization works and how many benefits the new method brings. In this section, the complexity of DMPA is analyzed in Subsection \ref{IV1}, which shows its advantage. Then Subsection \ref{IV2} analyzes the detection accuracy loss caused by discretization.
\subsection{Computation Complexity of DMPA}\label{IV1}
In each iteration of the MPA detection, messages are updated and sent along every edge in the factor graph. Equation (\ref{VNupdate}) shows that the complexity to update one message from layer nodes is $O(|\partial j|)$, which is negligible compared with (\ref{FNupdate}) and (\ref{convAll}). So the complexity orders of MPA and DMPA are determined by messages updating in resource nodes. This subsection first presents the complexity orders of MPA and DMPA. Then they are reanalyzed in the condition that the real part is independent with the imaginary part.

In Equation (\ref{FNupdate}), the part $y_k-h_{kj}x_{kjm}-\sum_{i\in \partial k\backslash j} h_{ki}c_{ik}$ and $\prod_{i\in \partial k} V_{i\rightarrow k}^{(t-1)}(\textbf{c}_i)$ are precomputed and stored. Then each item of the summation has complexity $O(1)$ and therefore the whole summation is $O(M^{d_f-1})$. Since the message updating in each edge needs to apply (\ref{FNupdate}) on $M$ codewords and each resource node has $d_f$ edges, the complexity order of MPA in each resource node is $d_f M^{d_f}$.

%y_k-h_{kj}x_{kjm}  -\sum_{i\in \partial k\backslash j} h_{ki}c_{ik}|^2] \prod_{i\in \partial k\backslash j} V_{i\rightarrow k}^{(t-1)}(\textbf{c}_i)

The complexity order of DMPA is determined by (\ref{convAll}). The length of the sequences in (\ref{convAll}) after zero-padding is computed by (\ref{discLen}) and has the property:
\begin{equation}
  N < 2[2((d_f-1)wid+nWid)/w+1].
\end{equation}
The complexity order of (\ref{convAll}) is the same as the 2-D FFT for $N$-length sequence, i.e., $d_f^2 \ln d_f$. Therefore the whole complexity order of DMPA for the $d_f$ edges in each resource node is $d_f^3 \ln d_f$.

When the real part is independent with the imaginary part in the SCMA system, the two parts can be detected individually, which largely reduces the complexity. In individual detection, the size of codebooks becomes $\sqrt M$ in MPA, so the new complexity is $O(d_f M^{d_f/2})$. Meanwhile the 2-D FFT in DMPA becomes 1-D FFT and the complexity order is reduced to $d_f^2 \ln d_f$. Subsection \ref{III3} has presented one situation in which the two parts are independent. Actually, the independence can be obtained by two requirements:
\begin{enumerate}
  \item  The real part of codebooks is designed independent with the imaginary part. Then the Cartesian Product of the two parts produces the final codebooks.
  \item  Different layers have the same channel vectors.
\end{enumerate}
The second requirement is satisfied when all layers are transmitted from the same transmit point \cite{SCMA}. Therefore in downlink model, the codebooks can be designed to apply individual detection.

\subsection{Detection Accuracy of DMPA}\label{IV2}
Subsection \ref{III3} has demonstrated that the $U_{k\rightarrow j}^{(t)}(\textbf{x}_{jm})$ in (\ref{FNupdate}) is the value of $g(y_k-x_{kjm})$. The discretization and FFT can approximately compute $g(y_k-x_{kjm})$ with low complexity but also brings approximation error, which causes detection accuracy loss. In this subsection, the detection accuracy of DMPA is analyzed by considering the approximation error. Firstly, the error is proved to be $0$ when there is no shift in discretization. Secondly, the upper bounds of error are derived when the length of shift is not $0$ but limited by sampling interval. In previous two steps the variables are assumed to be on real field , and finally the analysis result is extended to complex field.
%With the decrease of sampling length, the approximation error decreases and the detection accuracy of DMPA will approach that of MPA.
%DMPA prefers discretization and FFT instead of brute-force to update messages from resource nodes. With discretization, the $U_{k\rightarrow j}^{(t)}(\textbf{x}_{jm})$ in (\ref{FNupdate}) is approximately computed by Equation (\ref{appFNupdate}).  It largely reduces the detection complexity but also brings detection accuracy loss. This subsection estimates the detection accuracy of DMPA by analyzing the approximation error. With the decrease of approximation error, the detection accuracy of DMPA will approach that of MPA.

 \textbf{The sequence $g$ in (\ref{convAll}) is the exact sampling sequence of $g(t)$ when the elements of codewords are all on sampling positions.} The elements are reflected as impulses in the corresponding PDFs $\{f_i(t)\}_{i=1}^{d_f}$. In discretization procedure (\ref{disc1}), the impulses are distributed to the nearest sampling points. It can be seen as shifts on the impulses and hence the resulting sequence $f_i$ is not the real sampling sequence of $f_i(t)$. But when the impulses are exactly on some sampling positions, the $g$ is  not an approximation but the real sampling sequence of $g(t)$.
\begin{mylem}\label{lem2}
  Assume $f_1(t)$ consists of some impulses in $[t_1,t_1+(N-1)w]$, and $f_2(t)$ is an arbitrary function which is negligible outside $[t_2,t_2+(M-1)w]$. $f_1$ and $f_2$ are the discrete sequences generated by discretization in ($\ref{disc1}$) and ($\ref{disc2}$). $w$ is the sampling interval and the impulses are all on sampling positions. Then $f_1*f_2$ is the sampling sequence of $(f_1*f_2)(t)$.
\end{mylem}
\begin{proof}
  The function $f_1(t)$ can be denoted by
  \begin{equation}
    f_1(t)=\sum_{n=0}^{N-1} p_n \delta (t-(t_1+n\cdot w)).
  \end{equation}
  The sequences $f_1$ and $f_2$ are
  \begin{equation}
    f_1[n]=p_n~~ (n=\{0,1,\ldots,N-1\}),
  \end{equation}
  and
  \begin{equation}
    f_2[n]=f_2(t_2+n\cdot w) ~~(n=\{0,1,\ldots,M-1\}).
  \end{equation}
  The convolution of $f_1(t)$ and $f_2(t)$ is
  \begin{equation}
    (f_1*f_2)(t)=\sum_{n=0}^{N-1} p_n f_2(t-(t_1+n\cdot w)).
  \end{equation}
  It follows that
  \begin{align}
    (f_1*f_2)(t_1+t_2+m\cdot w)= & \sum_{n=0}^{m} p_n f_2(t_2+(m-n)\cdot w))\notag \\
    = & \sum_{n=0}^{m} f_1[n]f_2[m-n] \notag \\
    = & (f_1*f_2)[m],
  \end{align}
  where $m=\{0,1,\ldots,N+M-2\}$.
  Therefore $f_1*f_2$ is the sampling sequence of $(f_1*f_2)(t)$.
\end{proof}
According to Lemma \ref{lem2}, $g=f_2*f_3*\ldots*f_{d_f}*\eta$ is the sampling sequence of $g(t)$ when the impulses in $f_i(x)$ are not shifted in discretization.

\textbf{The approximation error of $g(y_k-x_{kjm})$ is limited by the sampling interval.} The previous paragraphs have shown that the approximation error is caused by shifts in discretization. Suppose the sampling interval is $w$, then the impulses in $f_i(t)$ are shifted less than $w/2$. As is shown in Fig. \ref{conv1}, the position of each impulse in the right is the summation of $d_f-1$ positions from the left. Therefore, the $M^{d_f-1}$ noise PDFs in Fig. \ref{conv2} are shifted less than $(d_f-1)w/2$. Suppose the $n$th point in sequence $g$ is corresponding to the position $t_0$ in $g(t)$ and is the nearest sampling point of position $t_0'$. Let $g'[n]$, $g'(t_0')$ respectively denote the approximate solution of $g[n]$ and $g(t_0')$, which are computed by (\ref{convAll}) and (\ref{discToFunc}). It follows that
\begin{equation}
  g[n]=g(t_0)=\sum_{i=1}^{M^{d_f-1}} p_i \eta(t_0-t_i),
\end{equation}
and
\begin{equation}
  g'[n]=\sum_{i=1}^{M^{d_f-1}} p_i \eta(t_0-t_i-\Delta d_i),
\end{equation}
where $p_i, t_i,\Delta d_i$ are the integral, position and shift of the $i$th impulse in $(f_2*\ldots *f_{d_f})(t)$, respectively. Assume the noise is white Gaussian noise with variance $\sigma^2$, then $\eta(t)$ has the property that
\begin{equation}
  \left| \frac{d\eta}{dt}(t) \right|=\left|\frac{-x}{\sigma^3\sqrt{2\pi}} e^{-x^2/2\sigma^2}\right|\leq \frac{1}{\sigma^2\sqrt{2\pi}} e^{-1/2}.
\end{equation}
 The absolute difference of $g[n]$ and $g'[n]$ is
\begin{align}
   \left| g[n]-g'[n]\right|= &  \left|\sum_{i=1}^{M^{d_f-1}} p_i (\eta(t_0-t_i)-\eta(t_0-t_i-\Delta d_i))\right|\notag \\
    \approx & \left|\sum_{i=1}^{M^{d_f-1}} p_i \Delta d_i \frac{d\eta}{dt} (t_0-t_i) \right| \notag \\
    \leq & \left|\sum_{i=1}^{M^{d_f-1}} p_i \Delta d_i\right| \frac{1}{\sigma^2\sqrt{2\pi}} e^{-1/2} \notag \\
    \leq &  \frac{1}{2\sigma^2\sqrt{2\pi}} (d_f-1)w e^{-1/2}.
  \end{align}
  Similarly, the absolute difference of $g(t_0)$ and $g(t_0')$ satisfies
  \begin{align}
   \left| g(t_0)-g(t_0')\right|\leq & \left|\sum_{i=1}^{M^{d_f-1}} p_i (t_0-t_0') \right| \frac{1}{\sigma^2\sqrt{2\pi}} e^{-1/2} \notag \\
    \leq &  \frac{1}{2\sigma^2\sqrt{2\pi}}w e^{-1/2}.
  \end{align}
The DMPA finally returns $g'[n]$ as $g'(t_0')$ in Equation (\ref{discToFunc}), then the approximation error can be upper bounded.
\begin{mypro}[Upper Bound 1]
For an arbitrary $t_0'$, the approximation error satisfies
\begin{align}\label{absErr}
  |g'(t_0')-g(t_0')|\leq & |g'[n]-g[n]|+|g(t_0)-g(t_0')| \notag \\
  \leq & \frac{1}{2\sigma^2\sqrt{2\pi}} d_f w e^{-1/2}.
\end{align}
\end{mypro}

It is obvious that $g'(t_0')-g(t_0')$ tends to $0$ with the decrease of $w$. The approximation error can also be estimated in a relative way. The $\eta(t)$ has the property that
\begin{equation}
  \frac{d\eta}{dt}(t_0)=-\frac{t_0}{\sigma^2} \eta(t_0).
\end{equation}
It follows that
\begin{align}
  \frac{|g[n]-g'[n]|}{g(t_0')}\approx & \frac{|g[n]-g'[n]|}{g[n]} \notag\\
  \leq & \frac{\sum_{i=1}^{M^{d_f-1}} p_i \left|\Delta d_i \frac{d\eta}{dt}(t_0-t_i)\right|}{\sum_{i=1}^{M^{d_f-1}} p_i\eta(t_0-t_i)}.
\end{align}
Because $\eta(t)$ is negligible outside $[-nWid, nWid]$, $\eta(t_0-t_i)$ can be seen as $0$ when $|t_0-t_i|>nWid$. It means that
\begin{equation}
  \left|\Delta d_i \frac{d\eta}{dt}(t_0-t_i)\right|\leq \frac{(d_f\!-\!1)w}{2} \cdot \frac{nWid}{\sigma^2} \eta(t_0-t_i).
\end{equation}
Similarly, the relative error $|g(t_0)-g(t_0')|/g(t_0')$ can be upper bounded.
\begin{mypro}[Upper Bound 2]
The total relative error can be estimated as
\begin{equation}\label{relaErr}
  \frac{|g'(t_0')-g(t_0')|}{g(t_0')}\leq \frac{wd_f\cdot nWid}{2\sigma^2}.
\end{equation}
\end{mypro}
Given noise variance $\sigma^2$£¬ the sampling interval $w$ can be modified such that the approximation error in DMPA is very small compared with the original messages. For example, the $w$ can be set as $0.02\sigma^2/(d_f \cdot n)$, then $|g'(t_0')-g(t_0')| \leq 0.01 g(t_0')$. Note that the upper bounds in (\ref{absErr}) and (\ref{relaErr}) are not tight and generally the $w$ needn't to be so small to ensure accuracy.

The previous derivation can be extended to the complex field. When the impulses are on complex field, the shifts are on both real part and the imaginary part, so the lengths of shifts are less than $\sqrt 2(d_f-1)w/2$. The gradient of PDF of complex white Gaussian noise in (\ref{scma}) is
\begin{equation}
  \nabla \eta(x,y)=(-\frac{2x}{N_0^2\pi} \exp(-\frac{x^2\!+\!y^2}{N_0}), -\frac{2y}{N_0^2\pi} \exp(-\frac{x^2\!+\!y^2}{N_0})).
\end{equation}
It has the properties that
\begin{align}
  |\nabla \eta(x,y)|=&\frac{2}{N_0^2 \pi} \exp(-\frac{x^2+y^2}{N_0}) \sqrt{x^2+y^2} \notag \\
  \leq & \frac{2}{N_0^2 \pi} \sqrt{\frac{N_0}{2e}},
\end{align}
and
\begin{align}
  \frac{|\nabla \eta(x,y)|}{\eta(x,y)}=&\frac{2}{\pi^2 N_0^3} \sqrt{x^2+y^2} \notag \\
  \leq & \frac{2 nWid}{\pi^2 N_0^3},
\end{align}
where $\eta(x,y)$ is negligible when $\sqrt{x^2+y^2}>nWid$. Consequently, the approximation error in complex field satisfies that
\begin{equation}
  |g'(t_0')-g(t_0')|\leq \frac{d_f w}{N_0^2\pi}\sqrt{\frac{N_0}{e}},
\end{equation}
and
\begin{equation}
   \frac{|g'(t_0')-g(t_0')|}{g(t_0')}\leq \frac{\sqrt{2}d_f w\cdot nWid}{\pi^2N_0^3}.
\end{equation}

\section{Numerical Results}\label{V}
In this section, the detection accuracy and time consumption of DMPA are evaluated by software simulation. The results show that the block error rate (BLER) of DMPA with a proper sampling interval $w$ is nearly the same with that of original MPA. Meanwhile, its detection time is significantly less than the original MPA and LLR MPA.
Firstly, some necessary parameters in simulation are introduced in Subsection \ref{V1}. Then the simulation results for accuracy and time consumption are presented in Subsection \ref{V2}.
\subsection{Parameters}\label{V1}
\textbf{Factor Graph}\\
In the model in Fig. \ref{scmaFG}, every layer node has degree $2$ which is corresponding to a combination of two resource nodes. The factor graph in our simulation follows this rule and modifies the parameter $K$. Then it has the properties: $J={K\choose 2}$, $d_f=2J/K=K-1$.

\textbf{Codebooks}\\
In this simulation, the real part and imaginary part are designed independent. The two parts are subjectively chosen from the interval $[-1,1]$ and their Cartesian Product is the final codewords. The codebook size $M$ is set as $16$. Without theoretical optimization, the performance of the codebooks is not good, but it is enough to confirm the main focus of this paper.

\textbf{BLER}\\
The $\log_2^M$ bits from the users are seen as one block. The BLER is the average error rate of the detected blocks.

Besides, the elements of channel vectors are all $1$ and the iteration number in detection is $5$. The noise region parameter $nWid$ is set as $5$. The setting is reasonable because with the $N_0=0.2$ (the largest $N_0$ used in simulation) the noise PDF $\eta(t)$ is smaller than $e^{-50}$ outside the range. The real and imaginary parts are detected independently in this simulation to reduce time cost.
\subsection{Performance Comparison}\label{V2}
Subsection \ref{IV2} has shown that the detection accuracy of DMPA will approach that of MPA with the decrease of $w$. The theoretical result is confirmed by Fig. \ref{berComp}, in which the detection accuracies of MPA and DMPA are nearly the same when $w=0.05$. Note that the detection error rate of DMPA may increase when the noise variance decreases. The reason is that in DMPA the accuracy is affected by not only noise but also the discretization. If the variance is very small while the sampling interval is not small enough, the approximation error in discretization would be large as shown in (\ref{absErr}) and (\ref{relaErr}). As a result, $N_0=0.004$ has a better accuracy than $N_0=0.002$ in DMPA when $w=0.2,0.3$.

\begin{figure}[htbp]
  \centering
  \includegraphics[width=0.4\textwidth]{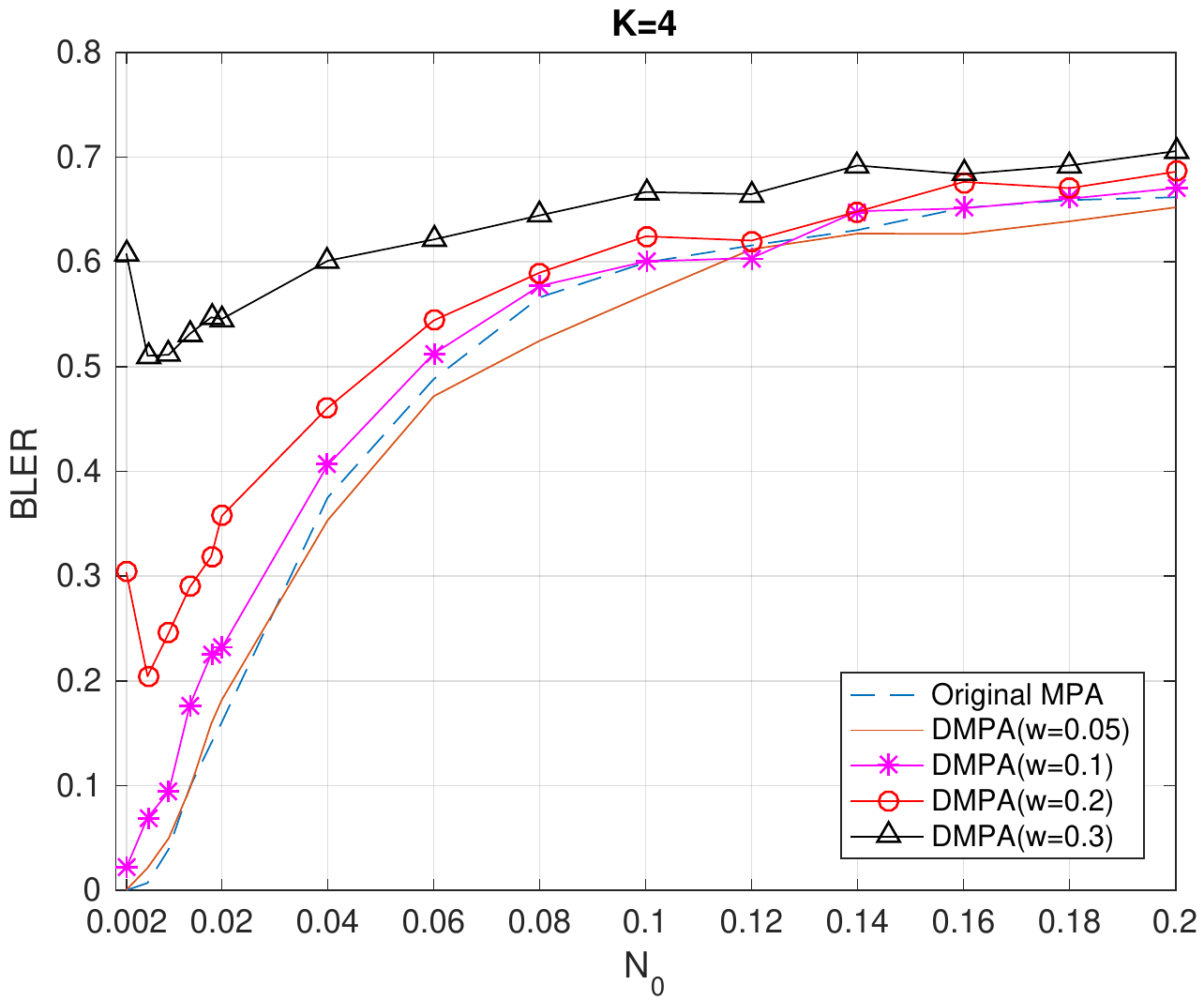}
  \centering
  \caption{Detection accuracy of MPA and DMPA with different sampling intervals}
  \label{berComp}
\end{figure}

The DMPA has been proved to have big advantage in complexity order compared with MPA, but the actual time consumption still need to be evaluated. In this simulation, three detection methods are run 100 times for different $d_f$ and the average time is presented in Table \ref{tCompTab},
\begin{table}[htbp]
\centering
\caption{Time consumption of different detection methods}\label{tCompTab}
\begin{tabular}{|c|c|c|c|c|}
  \hline
   & $d_f=2$ & $d_f=3$ & $d_f=4$ & $d_f=5$ \\
   \hline
  DMPA & 0.003518s & 0.010178s &  0.020442s & 0.037895s \\
  \hline
  LLR MPA & 0.004706s & 0.029063s & 0.232413s & 1.836131s \\
  \hline
  Original MPA & 0.004140s & 0.030238s & 0.241989s & 1.881747s \\
  \hline
\end{tabular}
\end{table}
which is also visually displayed in Fig. \ref{timeComp}. This figure shows that DMPA costs less time than both original MPA and LLR MPA especially when $d_f$ is large. Meanwhile the DMPA ($w=0.05$) has same accuracy with the original MPA, which is better than LLR MPA. In this simulation, the time consumption of LLR MPA is not significantly better than the original one, but in practice the hardware can handle addition operation much faster.
\begin{figure}[htbp]
  \centering
  \includegraphics[width=0.4\textwidth]{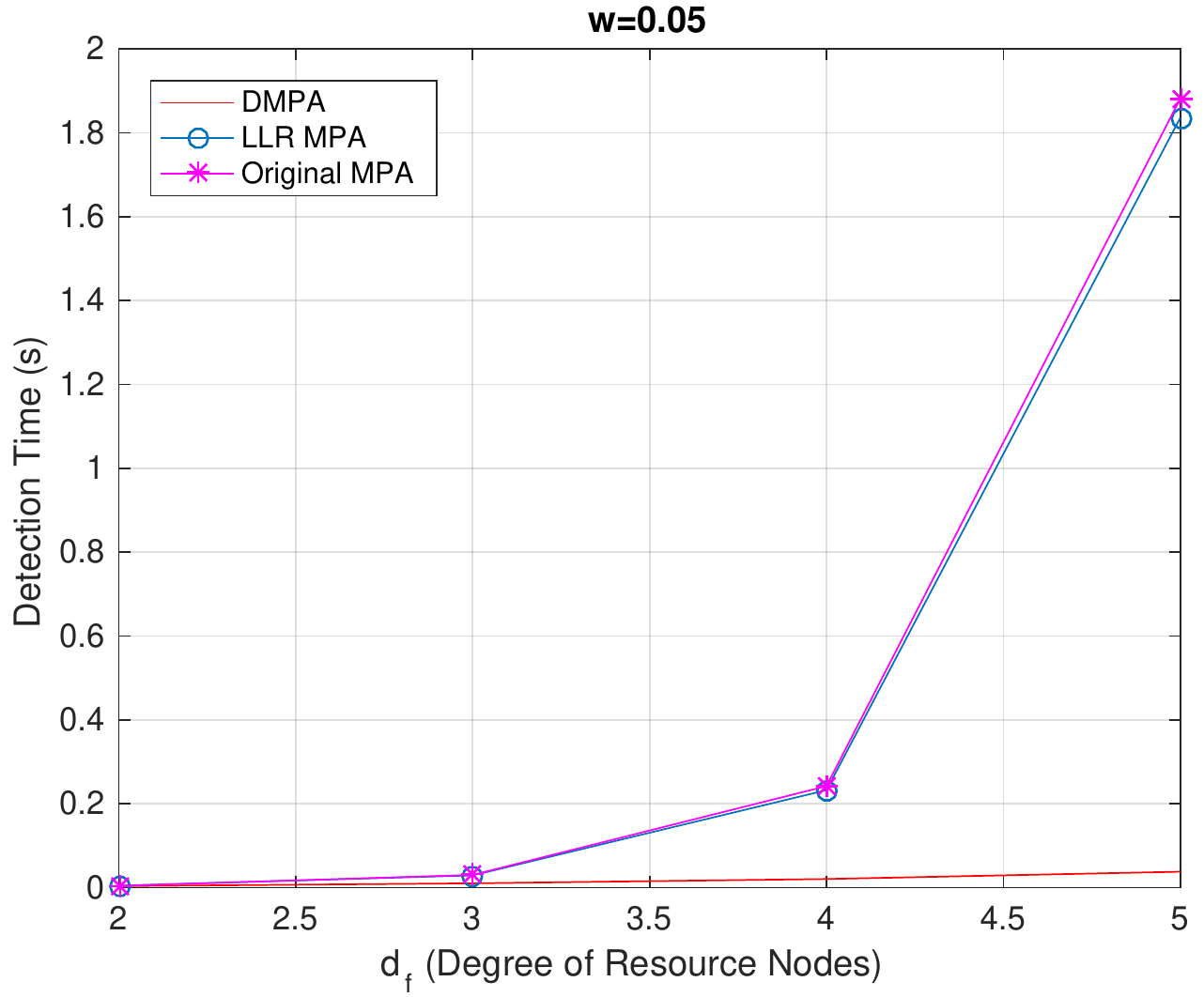}
  \centering
  \caption{Detection time of MPA and DMPA}
  \label{timeComp}
\end{figure}

\section{Conclusion}
In this paper, a low complexity detection algorithm called DMPA is proposed for SCMA. Instead of traversing on the Cartesian Product of $d_f-1$ codebooks, DMPA regards the layer nodes as random variables on complex field and hence applies discretization and FFT to update the messages. This process actually makes the $M^{d_f-1}$ impulses overlaying on the sampling points, which reduces the detection complexity per resource node from $O(d_f M^{d_f})$ to $O(d_f^3 \ln (d_f))$. The discretizaiton will cause approximation error and two upper bounds are derived in Section \ref{IV} to estimate the error. The upper bounds theoretically prove that the detection accuracy of DMPA will approach that of original MPA with the decrease of sampling interval $w$, which is confirmed by simulations. Numerical results show that compared with original MPA and LLR MPA, the DMPA has significant advantage on computation complexity when the accuracy loss is negligible.

%\section*{Acknowledgment}



%\bibliographystyle{IEEEtran}
%\bibliography{IEEEabrv,SCMA}

\end{document}